\documentclass[preprint,showpacs,preprintnumbers,amsmath,amssymb]{revtex4}
\usepackage{amsmath}

\usepackage{graphicx}
\usepackage{dcolumn}
\usepackage{bm}
\usepackage{amsmath}

\newtheorem{theorem}{Theorem}[section]
\newtheorem{lemma}[theorem]{Lemma}

\newenvironment{proof}[1][Proof]{\begin{trivlist}
\item[\hskip \labelsep {\bfseries #1}]}{\end{trivlist}}

\newcommand{\qed}{\nobreak \ifvmode \relax \else
      \ifdim\lastskip<1.5em \hskip-\lastskip
      \hskip1.5em plus0em minus0.5em \fi \nobreak
      \vrule height0.75em width0.5em depth0.25em\fi}

\begin{document}

\preprint{}

\title{Symmetric Halves of the $\frac{8}{33}$-Probability that the Joint State of Two Quantum Bits is Disentangled}
\author{Paul B. Slater}%
\email{slater@kitp.ucsb.edu}
\affiliation{%
University of California, Santa Barbara, CA 93106-4030\\
}%
\date{\today}

\begin{abstract}
Compelling evidence--though yet no formal proof--has been adduced that the probability that a generic 
two-qubit state ($\rho$) is separable is 
$\frac{8}{33}$ (arXiv:1301.6617, arXiv:1109.2560, arXiv:0704.3723). Proceeding in related analytical frameworks, using a further determinantal moment formula  of C. Dunkl 
(Appendix), 
we reach the conclusion that one-half of this probability arises when the determinantal inequality $|\rho^{PT}|>|\rho|$, where $PT$ denotes the partial transpose, is satisfied, and, the other half,  when $|\rho|>|\rho^{PT}|$. These probabilities are taken with respect to the flat, Hilbert-Schmidt measure on the fifteen-dimensional convex set of $4 \times 4$ density matrices. We find fully parallel bisection/equipartition results for the previously adduced, as well,  two-"re[al]bit" and two-"quater[nionic]bit"separability probabilities of $\frac{29}{64}$ and $\frac{26}{323}$, respectively. The computational results reported lend strong support to those obtained earlier--including the "concise formula" that yields them--most conspicuously amongst those findings  being the 
$\frac{29}{64}, \frac{8}{33}$ and $\frac{26}{323}$ probabilities noted.
\end{abstract}


\maketitle
The problem of determining the probability that generic sets of bipartite/multipartite quantum states exhibit entanglement features of one form or another is clearly of intrinsic interest \cite{ZHSL,simon,BHN,sbz,ingemarkarol}. We have reported \cite{MomentBased,slaterJModPhys} major advances, in this regard, with respect to the "separability/disentanglement probability" of two-qubit states, endowed with the flat, Hilbert-Schmidt (HS) measure \cite{szHS,ingemarkarol}. (The alternative use of the theoretically-important Bures [minimal monotone] measure 
\cite{szBures,ingemarkarol} has subsequently been investigated \cite{BuresHilbert,Hybrid}, but much less progress has so far been achieved in that area.) In particular, a concise formula \cite[eqs. (1)-(3)]{slaterJModPhys} 
\begin{equation} \label{Hou1}
P(\alpha) =\Sigma_{i=0}^\infty f(\alpha+i),
\end{equation}
where
\begin{equation} \label{Hou2}
f(\alpha) = P(\alpha)-P(\alpha +1) = \frac{ q(\alpha) 2^{-4 \alpha -6} \Gamma{(3 \alpha +\frac{5}{2})} \Gamma{(5 \alpha +2})}{3 \Gamma{(\alpha +1)} \Gamma{(2 \alpha +3)} 
\Gamma{(5 \alpha +\frac{13}{2})}},
\end{equation}
and
\begin{equation} \label{Hou3}
q(\alpha) = 185000 \alpha ^5+779750 \alpha ^4+1289125 \alpha ^3+1042015 \alpha ^2+410694 \alpha +63000 = 
\end{equation}
\begin{displaymath}
\alpha  \bigg(5 \alpha  \Big(25 \alpha  \big(2 \alpha  (740 \alpha
   +3119)+10313\big)+208403\Big)+410694\bigg)+63000
\end{displaymath}
has been developed that yields for a given 
$\alpha$, where $\alpha$ is a random-matrix-Dyson-like-index \cite{dumitriu}, the corresponding separability probability $P(\alpha)$. 

The setting $\alpha=1$ pertains to the fifteen-dimensional convex set of (standard, complex-entries) two-qubit density matrices, and the formula yields (to arbitrarily high numerical precision) $P(1) =\frac{8}{33}$ (cf. \cite{slater833} \cite[eq. B7]{joynt}). It is interesting to note that in this standard case \cite{steve}, the probability seems of a somewhat simpler nature (smaller numerators and denominators) than the value $P(\frac{1}{2}) =\frac{29}{64}$ obtained for the nine-dimensional convex set of $4 \times 4$ (two-"rebit") density matrices with real entries \cite{carl}, or, the value  $P(2) = \frac{26}{323}$ derived for the 
27-dimensional convex set of $4 \times 4$ (two-"quaterbit") density matrices with quaternionic entries \cite{asher2,adler}.

These simple rational-valued separability probabilities and the formula above that yields them were obtained through a number of distinct steps of analysis. First, based on extensive computations, C. Dunkl was able to obtain the (yet formally unproven) determinantal-moment formula \cite[p. 30]{MomentBased} (cf. \cite[eq. (28)]{zozor})
\begin{gather*} \label{nequalzero}
\left\langle \left\vert \rho^{PT}\right\vert ^{n}\right\rangle =\frac
{n!\left(  \alpha+1\right)  _{n}\left(  2\alpha+1\right)  _{n}}{2^{6n}\left(
3\alpha+\frac{3}{2}\right)  _{n}\left(  6\alpha+\frac{5}{2}\right)  _{2n}}\\
+\frac{\left(  -2n-1-5\alpha\right)  _{n}\left(  \alpha\right)  _{n}\left(
\alpha+\frac{1}{2}\right)  _{n}}{2^{4n}\left(  3\alpha+\frac{3}{2}\right)
_{n}\left(  6\alpha+\frac{5}{2}\right)  _{2n}}~_{5}F_{4}\left(
\genfrac{}{}{0pt}{}{-\frac{n-2}{2},-\frac{n-1}{2},-n,\alpha+1,2\alpha
+1}{1-n,n+2+5\alpha,1-n-\alpha,\frac{1}{2}-n-\alpha}%
;1\right) .
\end{gather*}
(The brackets denote expectation with respect to Hilbert-Schmidt [Euclidean] measure, while $5F4$ indicates a generalized hypergeometric function. The partial transpose of $\rho$, obtained by transposing in place its four $2 \times 2$ blocks, is denoted by 
$\rho^{PT}$.)
7,501 of these moments ($n =0,1,\ldots 7500$) were employed as input to a Mathematica program of Provost \cite[pp. 19-20]{Provost}, implementing a Legendre-polynomial-based-moment-inversion routine. From the high-precision, exact-arithmetic results obtained, we were able to formulate highly convincing, well-fitting conjectures (including the above-mentioned $\frac{8}{33}$ for $\alpha =1$) as to  underlying simple rational-valued separability probabilities. Then, with the use of the Mathematica FindSequenceFunction command applied to the sequence ($\alpha =1, 2,\ldots,32$) of these conjectures, and simplifying manipulations applied to the lengthy Mathematica result generated, we derived a multi-term $7F6$ hypergeometric-based formula \cite[Fig. 3]{slaterJModPhys} (cf. \cite[eq. (11)]{karolkarol}), with argument $\frac{27}{64}= (\frac{3}{4})^3$, for the conjectured values. Then, Qing-Hu Hou applied \cite[Figs. 5, 6]{slaterJModPhys} a highly celebrated ("creative telescoping") algorithm of Zeilberger \cite{doron} to this expression to obtain the concise separability probability formula ((\ref{Hou1})-(\ref{Hou3})) for $P(\alpha)$ itself.

In the course of his work in obtaining the $5F4$-hypergeometric-based HS moment formula above--and a more general one still for $\left\langle \left\vert \rho^{PT}\right\vert
^{n}\left\vert \rho\right\vert ^{k}\right\rangle /\left\langle \left\vert
\rho\right\vert ^{k}\right\rangle$--Dunkl employed certain "utility functions", in particular \cite[p. 26]{MomentBased},
\begin{align*}
F_{2}\left(  n,k\right)   &  =\left\langle \left\vert \rho\right\vert
^{k}\left(  \left\vert \rho^{PT}\right\vert -\left\vert \rho\right\vert
\right)  ^{n}\right\rangle /\left\langle \left\vert \rho\right\vert
^{k}\right\rangle. 
\end{align*}
Recently, upon request, he was able to obtain  the explicit formula (Appendix)
\begin{align*}
F_{2}\left(  n,k\right)   &  =
\frac{\left(-\frac{1}{16}\right)^n (\alpha )_n \left(\alpha
   +\frac{1}{2}\right)_n (2 k+n+5 \alpha +2)_n}{\left(k+3 \alpha
   +\frac{3}{2}\right)_n \left(2 k+6 \alpha +\frac{5}{2}\right)_{2 n}}\\
& \times~_{4}F_{3}\left(
\genfrac{}{}{0pt}{}{-\frac{n}{2},\frac{1-n}{2},k+1+\alpha,k+1+2\alpha
}{1-n-\alpha,\frac{1}{2}-n-\alpha,n+2k+2+5\alpha}%
;1\right)  .
\end{align*}
We set $k=0$ in this formula, and once again applied the Legendre-polynomial-based-moment-inversion procedure of Provost \cite{Provost}, in the same manner as in our previous studies. 

It was first necessary to note, however, that rather than the 
variable 
range $-\frac{1}{16} \leq |\rho^{PT}| \leq \frac{1}{256}$ employed in these earlier studies, the appropriate range would now be $-\frac{1}{16} \leq (|\rho^{PT}|-|\rho|) \leq \frac{1}{432}$. (Note that $432 = 2^4 \cdot 3^3$, as well as, of course, $16=2^4$ and $256 =2^8$.). The lower bound of $-\frac{1}{16}$ is achieved by Bell states--one example being the density matrix with $\frac{1}{2}$ in its four corners, and zeros elsewhere--and the upper bound of $\frac{1}{432}$ by a $4 \times 4$ density matrix with diagonal entries $\{\frac{1}{6},
\frac{1}{3},\frac{1}{3},\frac{1}{6}\}$ and (1,4) and (4,1)-entries equal to 
$-\frac{1}{6}$, and zeros otherwise. (Note that if we interchange the roles of 
$|\rho^{PT}|$ and $|\rho|$ in this last example, a value of $-\frac{1}{432}$, the lower bound on the domain of separability, is obtained for the variable 
$(|\rho^{PT}|-|\rho|)$ of interest.) We crucially rely throughout these series of analyses upon the  proposition that 
$|\rho^{PT}|>0$ is both a necessary and sufficient condition for a two-qubit state to be separable \cite{augusiak,Demianowicz}. (We note that the partial transpose of a $4 \times 4$ density matrix $\rho$ can possess at most one negative eigenvalue, so the 
non-negativity of $|\rho^{PT}|$--the product of the four eigenvalues of $\rho^{PT}$--is 
tantamount to separability.)

For the subrange $[0, \frac{1}{432}]$ of $(|\rho^{PT}|-|\rho|)$, 
containing only separable states, employing 
$\alpha =1$ in the new hypergeometric-based  formula immediately above, we obtained,
based on 9,451 ($n =0, 1,\ldots 9,450$) moments, an estimate that was 0.50000004358 as large as $\frac{8}{33}$.
The parallel calculations in the two-rebit ($\alpha = \frac{1}{2}$) and two-"quaterbit" 
($\alpha = 2$) cases yielded counterpart estimates of 0.5000025687 and 0.5000000000177, respectively. (Differences in rates of convergence--much the same as observed in \cite{MomentBased}--can be attributed to the initial [zero{\it th}-order] assumption of the Legendre-polynomial-moment-inversion procedure that the probability distributions to be fitted are uniform in nature, rendering more sharply-peaked distributions more difficult to rapidly approximate well. {\it A fortiori}, for the $\alpha =4$ (conjecturally octonionic) value \cite[p. 9]{slaterJModPhys}, $P(4)= \frac{4482}{4091349}$, the computed value here was, $0.500000000000000015 \times P(4).$) These outcomes, certainly, help to strongly bolster the validity of the (yet formally unproven) concise formula of Hou ((\ref{Hou1})-(\ref{Hou3})), yielding the full generic Hilbert-Schmidt two-qubit separability probabilities $P(\alpha)$.

For the two-rebit, two-qubit and two-quaterbit probabilities over the extended interval
$[-\frac{1}{432}, \frac{1}{432}]$, containing all separable and now some entangled states (and thus providing upper bounds on the total separability probabilities), the estimates, again based on 9,451 moments were 0.78082617689, 0.69244685258 and 0.601390039979. However, we did not discern any particular underlying common structure in these values. As examples of  entangled states dense in $[-\frac{1}{432}, 0]$, Dunkl advanced a one-parameter 
($s \in [-\frac{1}{108}, 0]$)) family of density matrices with diagonal entries 
$\frac{1}{2} -\frac{s}{2}, \frac{s}{2}, \frac{s}{2},\frac{1}{2} -\frac{s}{2}$,
and (1,4)- and (4,1)-entries equal to $s$, and zeros elsewhere.

In Fig.~\ref{fig:BisectionPlot} we display an estimate based on the first 51 $(n =0,\ldots,50)$ moments of the probability distributions in question as a function of $\alpha$ over the subrange $[-\frac{1}{108},\frac{1}{432}]$ of the full range $[-\frac{1}{16},\frac{1}{432}]$ of 
$(|\rho^{PT}|-|\rho|)$. The distributions are more sharply peaked for smaller $\alpha$ (nearer to $\alpha = \frac{1}{2}$ in the plot), as the larger values of 
$P(\alpha)$ for smaller $\alpha$ would indicate.
\begin{figure}
\includegraphics{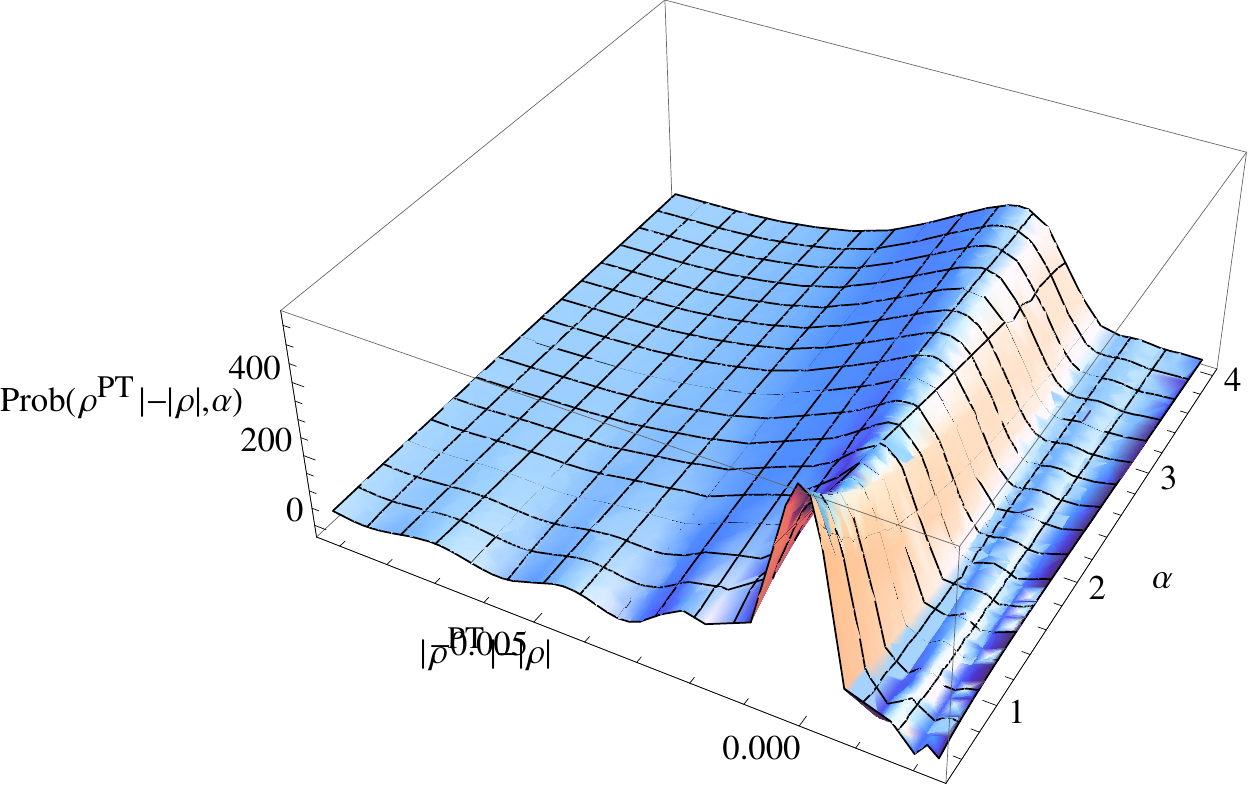}
\caption{\label{fig:BisectionPlot}Estimate based on first 51 moments of the probability distributions, as a function of the Dyson-index-like parameter $\alpha$, of the 
variable ($|\rho^{PT}|-|\rho|$)}
\end{figure}

Let us note that these "half-separability-probabilities" of $\frac{29}{128}, \frac{4}{33}, \frac{13}{323}$, indicated above, appear, by Hilbert-Schmidt-based analyses of Szarek, Bengtsson and {\.Z}yczkowski \cite{sbz}, to be exactly equal to the "full-separability-probabilities"
for the corresponding minimally-degenerate (boundary) generic two-rebit, two-qubit and 
two-quaterbit states (for which the determinant $|\rho|=0$). It would certainly be of interest to attempt to reproduce these probabilities--again employing $\alpha$ as a Dyson-index-like parameter--in a moment-based analysis (cf. \cite[App. D.7]{MomentBased}) analogous to that conducted above and previously in \cite{MomentBased,slaterJModPhys}.
(The range of $|\rho^{PT}|$ under the constraint $|\rho|=0$ would, first, have to be determined.)

These last three authors  had established that the set of positive-partial-transpose states for an arbitrary bipartite systems is "pyramid-decomposable" and hence, a body of constant height". They stated that "since our reasoning hinges directly on the Euclidean geometry, it does not allow one to predict any values of analogous ratios computed with respect to the Bures measure, nor other measures" \cite[p. L125]{sbz}.
Let us, then, pose the question of whether the "symmetric halves" finding elucidated above is itself particular to the Hilbert-Schmidt  (flat, Euclidean) metric or, 
contrastingly, does extend to the use of alternative metrics, such as the Bures (minimal monotone) metric \cite{szBures,ingemarkarol}? Also, in need of clarification is the issue of whether or not the Dyson-index {\it ansatz} of random matrix theory \cite{dumitriu}--apparently applicable in the Hilbert-Schmidt case, as our various results so far would indicate--extends to other measures, as well (cf. \cite{BuresHilbert,Hybrid}).
\section{Appendix (C. Dunkl)}
Let%
\[
g\left(  k,n\right)  :=\frac{\left(  k+1\right)  _{n}\left(  k+1+\alpha
\right)  _{n}\left(  k+1+2\alpha\right)  _{n}}{2^{6n}\left(  k+3\alpha
+\frac{3}{2}\right)  _{n}\left(  2k+6\alpha+\frac{5}{2}\right)  _{2n}},
\]
there is a multiplication relation:%
\[
g\left(  0,k\right)  g\left(  k,n\right)  =g\left(  0,k+n\right)  .
\]
Let%
\[
h\left(  k,n\right)  :=~_{5}F_{4}\left(
\genfrac{}{}{0pt}{}{-n,-k,\alpha,\alpha+\frac{1}{2},-2k-2n-1-5\alpha
}{-k-n-\alpha,-k-n-2\alpha,-\frac{k+n}{2},-\frac{k+n-1}{2}}%
;1\right)  .
\]
Then%
\begin{align*}
\left\langle \left\vert \rho\right\vert ^{k}\right\rangle  & =g\left(
0,k\right)  \\
\left\langle \left\vert \rho^{PT}\right\vert ^{n}\left\vert \rho\right\vert
^{k}\right\rangle /\left\langle \left\vert \rho\right\vert ^{k}\right\rangle
& =g\left(  k,n\right)  h\left(  k,n\right)  \\
\left\langle \left\vert \rho^{PT}\right\vert ^{n}\left\vert \rho\right\vert
^{k}\right\rangle  & =g\left(  0,k+n\right)  h\left(  k,n\right)  .
\end{align*}
Define%
\[
F_{2}\left(  n,k\right)  =\left\langle \left\vert \rho\right\vert ^{k}\left(
\left\vert \rho^{PT}\right\vert -\left\vert \rho\right\vert \right)
^{n}\right\rangle /\left\langle \left\vert \rho\right\vert ^{k}\right\rangle ,
\]
then%
\begin{align*}
F_{2}\left(  n,k\right)    & =\frac{1}{g\left(  0,k\right)  }\sum_{j=0}%
^{n}\binom{n}{j}\left(  -1\right)  ^{n-j}\left\langle \left\vert
\rho\right\vert ^{k+n-j}\left\vert \rho^{PT}\right\vert ^{j}\right\rangle \\
& =\frac{1}{g\left(  0,k\right)  }\sum_{j=0}^{n}\binom{n}{j}\left(  -1\right)
^{n-j}g\left(  0,k+n\right)  h\left(  k+n-j,j\right)  \\
& =g\left(  k,n\right)  \sum_{j=0}^{n}\binom{n}{j}\left(  -1\right)
^{n-j}h\left(  k+n-j,j\right)  .
\end{align*}
We will produce $F_{2}^{\prime}\left(  n,k\right)  :=\sum_{j=0}^{n}\binom
{n}{j}\left(  -1\right)  ^{n-j}h\left(  k+n-j,j\right)  $ as a single sum (so
that $F_{2}\left(  n,k\right)  =g\left(  k,n\right)  F_{2}^{\prime}\left(
n,k\right)  $).

\begin{lemma}
Let $n,m=0,1,2,\ldots$ and let $x$ be a variable, if $0\leq m\leq n$ then%
\[
\sum_{j=0}^{n}\frac{\left(  -n\right)  _{j}}{j!}\left(  -j\right)  _{m}\left(
x+j\right)  _{m}=\left(  -1\right)  ^{m}\frac{\left(  x\right)  _{2m}}{\left(
x\right)  _{n}}\left(  -n\right)  _{m}\left(  -m\right)  _{n-m},
\]
otherwise the sum is zero.
\end{lemma}

\begin{proof}
If $m>n$ then $\left(  -j\right)  _{m}=0$ for $0\leq j\leq n$. Suppose $0\leq
m\leq n$ then $\left(  -j\right)  _{m}=0$ for $0\leq j<m$ and the sum is over
$m\leq j\leq n$. Thus
\begin{align*}
\sum_{j=m}^{n}\frac{\left(  -n\right)  _{j}}{j!}\left(  -j\right)  _{m}\left(
x+j\right)  _{m}  & =\left(  -1\right)  ^{m}\sum_{j=m}^{n}\frac{\left(
-n\right)  _{j}~j!}{j!\left(  j-m\right)  !}\frac{\left(  x\right)
_{j}\left(  x+j\right)  _{m}}{\left(  x\right)  _{j}}\\
& =\left(  -1\right)  ^{m}\left(  x\right)  _{m}\sum_{j=m}^{n}\frac{\left(
-n\right)  _{j}\left(  x+m\right)  _{j}}{\left(  j-m\right)  !\left(
x\right)  _{j}}.
\end{align*}
Change the index of summation $j=m+i$ then the sum equals%
\begin{align*}
& \left(  -1\right)  ^{m}\frac{\left(  -n\right)  _{m}\left(  x\right)
_{m}\left(  x+m\right)  _{m}}{\left(  x\right)  _{m}}\sum_{i=0}^{n-m}%
\frac{\left(  m-n\right)  _{i}\left(  x+2m\right)  _{i}}{i!\left(  x+m\right)
_{i}}\\
& =\left(  -1\right)  ^{m}\frac{\left(  -n\right)  _{m}\left(  x\right)
_{2m}}{\left(  x\right)  _{m}}\frac{\left(  -m\right)  _{n-m}}{\left(
x+m\right)  _{n-m}}\\
& =\left(  -1\right)  ^{m}\left(  -n\right)  _{m}\left(  -m\right)
_{n-m}\frac{\left(  x\right)  _{2m}}{\left(  x\right)  _{n}},
\end{align*}
by the Chu-Vandermonde sum.
\end{proof}

Observe that $\left(  -m\right)  _{n-m}=0$ for $2m<n$. Then%
\begin{align*}
F_{2}^{\prime}\left(  n,k\right)    & =\left(  -1\right)  ^{n}\sum_{j=0}%
^{n}\frac{\left(  -n\right)  _{j}}{j!}\\
& \times\sum_{i=0}^{n}\frac{\left(  -j\right)  _{i}\left(  j-k-n\right)
_{i}\left(  \alpha\right)  _{i}\left(  \alpha+\frac{1}{2}\right)  _{i}\left(
-2k-2n-1-5\alpha\right)  _{i}}{i!\left(  -k-n-\alpha\right)  _{i}\left(
-k-n-2\alpha\right)  _{i}\left(  -\frac{k+n}{2}\right)  _{i}\left(
-\frac{k+n-1}{2}\right)  _{i}}\\
& =\left(  -1\right)  ^{n}\sum_{i=0}^{n}\frac{\left(  \alpha\right)
_{i}\left(  \alpha+\frac{1}{2}\right)  _{i}\left(  -2k-2n-1-5\alpha\right)
_{i}}{i!\left(  -k-n-\alpha\right)  _{i}\left(  -k-n-2\alpha\right)
_{i}\left(  -\frac{k+n}{2}\right)  _{i}\left(  -\frac{k+n-1}{2}\right)  _{i}%
}\\
& \times\sum_{j=0}^{n}\frac{\left(  -n\right)  _{j}}{j!}\left(  -j\right)
_{i}\left(  j-k-n\right)  _{i}.
\end{align*}
Apply the lemma to the $j$-sum with $x=-k-n$ and $m=i$ to obtain%
\[
\left(  -1\right)  ^{i}\left(  -n\right)  _{i}\left(  -i\right)  _{n-i}%
\frac{\left(  -n-k\right)  _{2i}}{\left(  -n-k\right)  _{n}}=\left(
-1\right)  ^{i}\frac{\left(  -n\right)  _{i}\left(  -i\right)  _{n-i}}{\left(
-n-k\right)  _{n}}2^{2i}\left(  -\frac{k+n}{2}\right)  _{i}\left(
-\frac{k+n-1}{2}\right)  _{i}%
\]
and thus
\[
F_{2}^{\prime}\left(  n,k\right)  =\frac{\left(  -1\right)  ^{n}}{\left(
-n-k\right)  _{n}}\sum_{i=0}^{n}\frac{\left(  -n\right)  _{i}\left(
-i\right)  _{n-i}\left(  \alpha\right)  _{i}\left(  \alpha+\frac{1}{2}\right)
_{i}\left(  -2k-2n-1-5\alpha\right)  _{i}}{i!\left(  -k-n-\alpha\right)
_{i}\left(  -k-n-2\alpha\right)  _{i}}\left(  -1\right)  ^{i}2^{2i}.
\]
This is not in hypergeometric form because of the term $\left(  -i\right)
_{n-i}$; also the summation extends over $\frac{n}{2}\leq i\leq n$. Change the
index $j=n-i$ then
\begin{align*}
\frac{\left(  -n\right)  _{i}}{i!}\left(  -i\right)  _{n-i}  & =\left(
-1\right)  ^{i}\frac{n!}{\left(  n-i\right)  !i!}\left(  -1\right)
^{n-i}\frac{i!}{\left(  2i-n\right)  !}=\left(  -1\right)  ^{n}\frac
{n!}{j!\left(  n-2j\right)  !}\\
& =\left(  -1\right)  ^{n}\frac{2^{2j}}{j!}\left(  -\frac{n}{2}\right)
_{j}\left(  \frac{1-n}{2}\right)  _{j}%
\end{align*}
and the reversal formula is%
\begin{align*}
\left(  x\right)  _{i}  & =\left(  x\right)  _{n-j}=\frac{\left(  x\right)
_{n-j}\left(  x+n-j\right)  _{j}}{\left(  x+n-j\right)  _{j}}\\
& =\left(  -1\right)  ^{j}\frac{\left(  x\right)  _{n}}{\left(  1-n-x\right)
_{j}}.
\end{align*}
Thus%
\begin{align*}
F_{2}^{\prime}\left(  n,k\right)    & =\frac{\left(  -1\right)  ^{n}\left(
\alpha\right)  _{n}\left(  \alpha+\frac{1}{2}\right)  _{n}\left(
-2k-2n-2k-1-5\alpha\right)  _{n}}{\left(  -n-k\right)  _{n}\left(
-k-n-\alpha\right)  _{n}\left(  -k-n-2\alpha\right)  _{n}}\\
& \times\sum_{j=0}\frac{\left(  -\frac{n}{2}\right)  _{j}\left(  \frac{1-n}%
{2}\right)  _{j}\left(  k+1+\alpha\right)  _{j}\left(  k+1+2\alpha\right)
_{j}}{j!\left(  1-n-\alpha\right)  _{j}\left(  \frac{1}{2}-n-\alpha\right)
_{j}\left(  n+2k+2+5\alpha\right)  _{j}}2^{2j+2n-2j}\\
& =\left(  -1\right)  ^{n}2^{2n}\frac{\left(  \alpha\right)  _{n}\left(
\alpha+\frac{1}{2}\right)  _{n}\left(  n+2k+2+5\alpha\right)  _{n}}{\left(
k+1\right)  _{n}\left(  k+1+\alpha\right)  _{n}\left(  k+1+2\alpha\right)
_{n}}\\
& \times~_{4}F_{3}\left(
\genfrac{}{}{0pt}{}{-\frac{n}{2},\frac{1-n}{2},k+1+\alpha,k+1+2\alpha
}{1-n-\alpha,\frac{1}{2}-n-\alpha,n+2k+2+5\alpha}%
;1\right)  ;
\end{align*}
a balanced sum.

The formula was tested for $F_{2}\left(  2,k\right)  $, also directly verified
for $n=3$, arbitrary $\alpha$.

Combining the front factors in $F_{2}\left(  n,k\right)  $ (from $g\left(
k,n\right)  $) we obtain%
\[
\left(  -1\right)  ^{n}\frac{\left(  \alpha\right)  _{n}\left(  \alpha
+\frac{1}{2}\right)  _{n}\left(  n+2k+2+5\alpha\right)  _{n}}{2^{4n}\left(
k+3\alpha+\frac{3}{2}\right)  _{n}\left(  2k+6\alpha+\frac{5}{2}\right)
_{2n}}.
\]

\begin{acknowledgments}
I would like to express appreciation to the Kavli Institute for Theoretical
Physics (KITP) for computational support in this research. C. Dunkl supplied the crucial $4F3$-hypergeometric-based formula employed, as well as much useful advice.
\end{acknowledgments}

\bibliography{Dissection}

\end{document}